\documentclass[12pt]{llncs}

\usepackage[utf8]{inputenc}
\usepackage[T1]{fontenc}
\usepackage[english]{babel}
\usepackage{color}
\usepackage{xspace}
\usepackage{tikz}
\usepackage{times}
\usepackage{placeins}

\usepackage[lined,boxed,linesnumbered,noend]{algorithm2e} 
\usepackage{amsmath,amsfonts,amssymb}
\usepackage[margin=1.35in]{geometry}
\usepackage{multicol}
\usepackage{multirow}
\usepackage{color}

\newcommand{\G}{\mathcal{G}}

\newcommand{\Pred}{\mathcal{P}}
\newcommand{\PredAdd}{\mathcal{P}^{+}}
\newcommand{\FTS}{\ensuremath{\G^*_{st}}\xspace}
\newcommand{\FT}{\ensuremath{\G^*}\xspace}
\newcommand{\leadstoo}{\overset{st}{\leadsto}}
\newcommand{\TVG}{\ensuremath{{\G=(V,E,\mathcal{T},\rho,\zeta)}}\xspace}

\newcommand{\fixme}[1]{{\color{red}{(fixme)}}\xspace}

\title{\hspace{-.5cm}Testing Temporal Connectivity in Sparse Dynamic Graphs\thanks{A short version appeared in French in ALGOTEL'14 \cite{BCCJN14}. This work is partially supported by the DGA through a PhD scholarship (No 2013 60 0074).}\hspace*{-1cm}}

\author{\large Matthieu Barjon,
  Arnaud Casteigts,
  Serge Chaumette,\\
  Colette Johnen,
  Yessin M. Neggaz
}

  \institute{LaBRI, University of Bordeaux}

\begin{document}

\maketitle

\begin{abstract}

We address the problem of testing whether a given dynamic graph is temporally connected, {\it i.e.} a temporal path (also called a {\em journey}) exists between all pairs of vertices. We consider a discrete version of the problem, where the topology is given as an evolving graph $\G=\{G_1,G_2,...,G_{k}\}$ whose set of vertices is invariant and the set of (directed) edges varies over time. Two cases are studied, depending on whether a single edge or an unlimited number of edges can be crossed in a same $G_i$ (strict journeys {\it vs} non-strict journeys). 

In the case of {\em strict} journeys, a number of existing algorithms designed for more general problems can be adapted. We adapt one of them to the above formulation of the problem and characterize its running time complexity. The parameters of interest are the length of the graph sequence $k=|\G|$, the maximum {\em instant} density $\mu=max(|E_i|)$, and the {\em cumulated} density $m=|\cup E_i|$. Our algorithm has a time complexity of $O(k\mu n)$, where $n$ is the number of nodes. This complexity is compared to that of the other solutions: one is always more costly (keep in mind that is solves a more general problem), the other one is more or less costly depending on the interplay between instant density and cumulated density. The length $k$ of the sequence also plays a role. We characterize the key values of $k, \mu$ and $m$ for which either algorithm should be used. Our solution is relevant for sparse mobility scenario (e.g. robots or UAVs exploring an area) where the number of neighbors at a given time is low, though many nodes can be seen over the whole execution.

In the case of {\em non-strict} journeys, for which no algorithm is known, we show that some pre-processing of the input graph allows us to re-use the same algorithm than before. By chance, these operations happens to cost again $O(k\mu n)$ time, which implies that the second problem is not more difficult than the first.

Both algorithms gradually build the transitive closure of strict journeys ($\G^*_{st}$) or non-strict journeys ($\G^*$) as the edges are examined; these are {\em streaming} algorithms. They stop their execution whenever temporal connectivity is satisfied (or after the whole graph has been examined). A by-product of the execution is to make $\G^*_{st}$ and $\G^*$ available for further connectivity queries (in a temporal version), these queries being then reduced to simple adjacency tests in a static graph.
\end{abstract}

\section{Introduction}

Connected and mobile devices such as mobile phones, satellites, cars, or robots form highly dynamic networks in which connectivity between nodes evolves rapidly and continuously. Furthermore, the topology of such a network at a given time is generally not connected, and even extremely sparse in the case of exploration or surveillance scenarios~\cite{BMF+00,FMS13}, or when passive mobility is considered with humans or animals~\cite{JSS05,SRJB03}. However, even in these extreme cases, a form of connectivity arises over time and space, by means of delay tolerant communications, where messages are retained until an opportunity of transmission appears (mechanisms of type "{\em store-carry-forward"}). This type of connectivity is referred to as {\em temporal connectivity}.

In this paper, the problem that we study is to automatically test whether a given dynamic graph is temporally connected or not. In other words, we want to decide if there is a temporal path ({\em journey}) between every pair of nodes in the network. A key concept is that of {\em transitive closure} of journeys, introduced in~\cite{BF03}. This is a static directed graph (even if the dynamic graph is not itself directed) whose edges represent the potential journeys. From this structure, the membership of a given dynamic graph to several classes of graphs can be decided~\cite{CCF09}, and in particular to the class of temporally connected graphs (complete transitive closure). We address the computation of transitive closure in the case of strict journeys (\FTS) or non-strict journeys (\FT) given a dynamic graph $\G=\{G_1,G_2,...,G_{k}\}$. 

In the case of strict journeys, several algorithms can be adapted to compute \FTS. Several algorithms are given in~\cite{BFJ03}, each computing optimal journeys according to a given criterion (foremost, shortest, fastest). Any of these algorithms can be adapted to compute \FTS\ by using the appropriate parameters. Precisely, these algorithms compute the journeys taking into account the duration of edges crossing (latency) and the duration of each graph $G_i$ (timed evolving graphs~\cite{Fer02}). If we assign to each $G_i$ a unit duration that also corresponds to the duration of edges crossing, then, for a given source node, the result is the set of strict journeys. This algorithm has to be executed $n$ times, once from each vertex, in order to compute the transitive closure of the journeys. The most efficient of the three algorithms (foremost journeys) has an execution time of $O(m\log k + n\log n)$, hence a total time of $O(n(m\log k + n\log n))$.

\label{sec:paris}

An algorithm computing a generalization of the transitive closure of journeys was proposed in~\cite{WDCG12}. This generalization, called {\em dynamic reachability graph}, corresponds to a transitive closure of journeys parametrized by a starting date, a maximal duration of the journeys, and a traversal time for edges. It applies to dynamic graphs represented as TVGs~\cite{CFQS12} ({\em time-varying graph}), namely a quintuplet $\TVG$ where $\mathcal{T}$ is the temporal domain ($\mathbb{R}^+$ in the case of~\cite{WDCG12}) and $\rho$ and $\zeta$ are functions that determine the presence and the latency of a given edge at a given instant, respectively. This algorithm  can be used to compute $\G^*_{st}$ as follows: First create a TVG whose edges presence dates (function $\rho$) are all multiples of some unit value that also corresponds to the latency given by function $\zeta$ (here, a constant). 
Finally, the constraint on journeys duration is set to $+\infty$; the departure date is set to $0$; then $\G^*_{st}$ is obtained by executing the algorithm from~\cite{WDCG12} on the created TVG. Informally, the strategy of that algorithm is to compose reachability graphs incrementally over increasing periods of time, namely each graph covering $2^{i}$ time steps is obtained by composition of two graphs covering $2^{i-1}$ time steps (for $i$ from $1$ to $\log k$). The complexity of this algorithm is $O(k \log k\ mn\log n)$. The authors do not exclude the possibility to get rid of the trailing $\log n$ factor, potentially linked to an implementation choice (see Section~4.3 of~\cite{WDCG12}). Either way, the complexity of that solution dominates that of our solution.

We propose a decicated approach for computing the transitive closure (strict at first) of an untimed directed evolving graph $\G=\{(V,E_i)\}$ which has a better time complexity than the adaptation of~\cite{WDCG12} in all cases, and than the adaptation of~\cite{BFJ03} for a range of dynamic graphs, in particular those whose density is low at any time, though arbitrarily dense over time. The algorithm consists of a temporal adaptation of the Bellman-Ford principle used in static graphs to compute distances between nodes. This principle was also adapted in~\cite{KKW08} to compute time lags between entities based on a contact history (e.g. a sequence of dated emails). Our adaptation is quite straight and its time complexity is $O(k \mu n)$ with the considered data structure (a mere sequence of sets of edges), where $k=|\G|$ is the length of the sequence (also called number of steps) and $\mu=max(|E_i|)$ is the maximal number of edges that exist at any given step. This last parameter is to be contrasted with $m=|\cup E_i|$, the total number of edges that exist over time. As discussed above, the distinction is relevant in a number of scenarios based on mobile communicating entities. Furthermore, this type of graphs typically corresponds to those in which the question of the temporal connectivity occurs, since it is not \emph{a priori} granted.
In the case of non-strict journeys, for which we do not know any existing algorithm, we show that the same solution can be directly adapted with the same time complexity: $O(k\mu n)$. This variant is based on a double transitive closure: a {\em static} transitive closure applied to each $G_i$ independently, and a temporal one (as in the case of strict journeys) applied to the sequence of those static closures.

Both algorithms gradually build the transitive closure as the edges are examined; these are {\em online} algorithms. They stop their execution as soon as the temporal connectivity is satisfied (or after the whole dynamic graph has been examined). A by-product of the execution is to make $\G^*_{st}$ and $\G^*$ available for further connectivity queries (in a temporal version), these queries being then reduced to simple adjacency tests in a static graph.

The rest of this paper is organized as follows. The key concepts and main notations are introduced in Section~\ref{sec:model}. Section~\ref{sec:strict} presents our solution to compute the transitive closure of strict journeys, which is then adapted to the case of non-strict journeys in Section~\ref{sec:non-strict}. Time complexity is analyzed throughout the paper. We provide in Section~\ref{sec:comparaison} a more detailed comparison that indicates the values of $k,\mu$ and  $m$ for which our solution performs better than the adaptation from~\cite{BFJ03}.

\section{Model and notations}
\label{sec:model}

Let $\G$ be an untimed directed evolving graph $\{G_i=(V,E_i)\}$. There is a {\em non-strict} journey from $u$ to $v$ in $\G$ if and only if there exists a sequence of edges $e_1, e_2, ..., e_p$ connecting $u$ to $v$ such that for all $j\in 1..p$$-$$1$, $e_{j}\in E_i \implies \exists i'\ge i, e_{j+1}\in E_{i'}$. If the inequality $i'>i$ is strict, the journey is called {\em strict journey}, i.e. at most one edge can be crossed in a single step $i$ (as opposed to an unlimited number for non-strict journeys). The existence of a non-strict (resp. strict) journey from $u$ to $v$, when the context is implicit, is noted $u\leadsto v$ (resp. $u\leadstoo v$). The distinction between a strict journey and a non-strict journey was introduced in~\cite{CCF09} to report on necessary or sufficient conditions on the dynamics of the graph, regarding distributed algorithms based on pairwise interactions. It should not be mistaken with the notion of direct or indirect journey~\cite{CFGSY13}, which corresponds to the very fact of allowing pauses in between consecutive hops.

The non-strict transitive closure of the dynamic graph $\G$ is the {\em static} directed graph $\FT=(V,E^*)$ such that $(u,v)\in E^* \Leftrightarrow u\leadsto v$. We define, in the same way, the strict transitive closure of $\G$ by the static directed graph $\FTS=(V,E_{st}^*)$ where $(u,v)\in E_{st}^* \Leftrightarrow u\leadstoo v$.
It should be noted that the transitive closure graph of $\G$ is directed whatever the nature (directed or not) of the edges in $\G$. This is due to the temporal dimension that, by nature, implies an orientation.

Given a dynamic graph $\G$, we note $k=|\G|$ the number of time steps in $\G$, i.e. the number of static graphs contained in $\G$. We distinguish two parameters to report on the number of edges in the graph: the maximal number of edges that exist at each single step, i.e. $\mu=max(|E_i|)$, and the total number of edges that exist over time, i.e. $m=|\cup E_i|$. Of course, whatever the considered graph, we have $m\ge \mu$. 
Moreover, as already discussed, it is not rare that a practical scenario verifies $\mu=o(n)$, or even $\mu=\Theta(1)$, while $m=\Theta(n\log n)$ or $m=\Theta(n^2)$.

\section{Computation of the transitive closure for strict journeys}
\label{sec:strict}

We propose below an algorithm for computing the strict transitive closure $\FTS$ in the general case that $\G$ is directed. The principle of the algorithm is to build, step by step, the list of all the predecessors of each vertex $v$, i.e., the set $\{u : u\leadstoo v\}$. Each step of the algorithm works on a static graph of $\G$. Let $\Pred(v,t)$ be the set of known predecessors of $v$ by the end of the  $t$ first steps of the algorithm (i.e. after taking into account edge sets: $E_1, ..., E_t$). The core of step $i$ is to add $\Pred(u,i-1)$ to $\Pred(v,i)$ for each edge $(u,v)\in E_i$. In practice, only two variables $\Pred(v)$ and $\PredAdd(v)$ are maintained on each node $v$. $\PredAdd(v)$ contains the new predecessors of $v$ (computed during the current step). At the end of the current step $\PredAdd(v)$ is merged to $\Pred(v)$, the set of all predecessors of $v$. 
The detailed operations are given in Algorithm~\ref{algo:FT-strict}.

\begin{algorithm}[h]
  \SetKwData{Old}{predecessors}\SetKwData{Working}{$\Pred$}\SetKwData{Successors}{successors}
  \SetKwFunction{Add}{Add}\SetKwFunction{Copy}{copy}\SetKwFunction{AddAll}{AddAll}
  \SetKwInOut{Input}{Input}\SetKwInOut{Output}{Output}
  \Input{A dynamic graph $\G$ given as $(V,\{E_i\})$}
  \Output{A set of edges $E^*$ such that $\FTS=(V,E^*)$}
  \BlankLine
  \tcp{Initialization}
  \ForEach{$v$ in $V$}{
    $\Pred(v) \gets \{v\}$\tcp*{Each node is its own predecessor}
    $\PredAdd(v) \gets \emptyset$\;
  }

  \ForEach{$E_i$ in $\{E_i\}$}{
    $UpdateV \gets \emptyset$ \tcp*{List of nodes whose predecessors will be updated}
    \tcp{List predecessors induced by the edges in $E_i$}
    \ForEach{$(u,v)$ in  $E_{i}$}{ 
      $\PredAdd(v) \gets \PredAdd(v) \cup \Pred(u)$\;
      $UpdateV \gets UpdateV \cup \{v\}$
    }
    
    \tcp{Add found predecessors to known predecessors}
    \ForEach{$v$ in $UpdateV$}{
      $\Pred(v) \gets \Pred(v) \cup \PredAdd(v)$\;
      $\PredAdd(v) \gets \emptyset$\;
    }

    \tcp{Test whether transitive closure is complete; if so, terminates}
    $isComplete \gets true$\;
    \ForEach{$v$ in $V$}{
      \If{$|\Pred(v)| < |V|$}{
        $isComplete \gets false$\;
        $break$\;
      }
    }
    \If{$isComplete$}{ \tcp{The algorithm terminates returning a complete graph (edges)}
    \KwRet{$V\times V\setminus \{loops\}$}
   } 
  }

  \tcp{Build transitive closure based on predecessors}
  $E^* \gets \emptyset$

  \ForEach{$v$ in $V$}{
    \ForEach{$u$ in $\Pred(v)\setminus \{v\}$}{
      $E^* \gets E^* \cup (u,v)$
    }
  }
  \KwRet{$E^*$}

  \caption{\label{algo:FT-strict}Computation of the strict transitive closure $\FTS$}
\end{algorithm}

\subsection{Time complexity}

This section provides an analysis of the time complexity of the algorithm. In this analysis, we consider the use of {\em set} data structures, for which the union has at worst a linear cost in the number of items of the two considered sets. We also assume that the size of a set can be known in constant time, which is the case with most of the existing libraries that implements this kind of data structure (this value being maintained as the set is modified).

\begin{lemma}
  \label{lem:nbpred}
  For all $v \in V$, $|\Pred(v)|\le k\mu$, i.e., a node cannot have more than $k\mu$ predecessors.
\end{lemma}

\begin{proof}[by contradiction]
  If there is a node $v$ such that $|\Pred(v)\setminus v|> k\mu$, then, by definition, there exists more than $k\mu$ vertices $u$ different from $v$ such that $u\leadsto v$. Each of these vertices is thus the origin of at least one edge, which means that more than $k\mu$ distinct edges existed.\qed
\end{proof}

\begin{theorem}
  Algorithm 1 that calculates the strict transitive closure of a graph $\G$ has a time complexity in $O(k \mu n)$.
\end{theorem}

\begin{proof}
The initialization loop is linear in $n$. The main loop iterates as many times as the number of steps in $\G$, i.e. $k$ times. The main loop has three sub-loops, each being dominated by $O(|E_i|\cdot n)=O(\mu n)$. Finally, the construction of the transitive closure, if it is not complete before the end, consists of a loop that, for each node, iterates over its predecessors. Since the number of predecessor of a given node cannot exceed $k\mu$ (Lemma~\ref{lem:nbpred}), this latter loop is also dominated by $O(k\mu n)$.\qed
\end{proof}

\section{Computation of the transitive closure for non-strict journeys}
\label{sec:non-strict}

In this section, we focus on the calculation of $\FT$, i.e. the transitive closure of the journeys for which an unlimited number of edges can be crossed at each step (non-strict journey). 
A simple observation allows us to reuse Algorithm~\ref{algo:FT-strict} almost directly. Indeed, the relaxation of the constraint that the journeys are strict implies that at each step $i$, if a path (in the classic acceptance of the word) exists from $u$ to $v$, then $u$ can join $v$ at the same step. The algorithm therefore consists in pre-computing, at each step, the transitive closure (in the classic static meaning of this term) of the edges present in $G_i$, resulting in a graph $G_i^*$, each edge of which corresponds to a path in $G_i$. Then Algorithm~\ref{algo:FT-strict}, applied to the dynamic graph $\{G_i^*\}$, produces directly the non-strict transitive closure $\FT$.

\subsection{Time Complexity}

The time complexity of this algorithm essentially depends on the cost of the calculation of the static transitive closure $G_i^*$ of the graphs $G_i$. This can be done by a depth first search (DFS) or by a breadth first search (BFS) run from each vertex in $G_i$. Each of these runs having an execution cost in $O(|E_i|)=O(\mu)$ and thus, the extra cost of this operation remains within $O(k \mu n)$ time.

\FloatBarrier
\section{Comparison}
\label{sec:comparaison}

This section compares the complexity of the proposed algorithm with the adaptation of that from~\cite{BFJ03} (based on foremost journeys), which has a running time of $O(n(m\log k + n\log n))$, where $m \ne \mu$ is the total number of edges existing over time, i.e. $|\cup E_i|$. 

The question is therefore to compare this complexity to $O(k\mu n)$, or after simplification by $n$, to compare $O(k\mu)$ to $O(m\log k + n\log n)$. These complexities belong to a four-dimensional space : $\mu, m, k$ and $n$; it is therefore not easy to compare them. We propose to study them asymptotically in $n$, by varying the values of $\mu, m$ and $k$. 
Precisely, we vary the order of $\mu$ and $m$ (instant density\  {\it vs.} cumulated density) for several ratios of possibles values of $k$ and $n$ (i.e. the length of the sequence $\G$ in function of $n$). Table~\ref{fig:tableau} contains 60 results, including a dozen that show the transition in efficiency between both solutions (the others can be extrapolated without calculation by considering the relative impact of factors $k$ and $\log k$ in both formulas). To make the verification of these results simpler, we provide in the right column an intermediate expression, obtained after replacing $\mu$ and $m$ in both expressions $O(k\mu)$ and $O(m\log k + n\log n)$.

\begin{table}[h]
  \centering
\begin{tabular}{|c|c||c|c|c|c||c|}
  \hline
  \multirow{2}{*}{$\mu=\Theta(.)$}&\multirow{2}{*}{$m=\Theta(.)$}&\multirow{2}{*}{$k=\Theta(\log n)$}&\multirow{2}{*}{$k=\Theta(\sqrt n)$}&\multirow{2}{*}{$k=\Theta(n)$}&\multirow{2}{*}{$k=\Theta(n^2)$}&Intermediate calculation\\
  & & & & & &$\Theta(.) \pm \Theta(.)$\\\hline
$\log n$&$n$&~&n/a&$\approx$&$+$&$k\log n \pm n\log k+n\log n$\\
$\sqrt n$&$n$&n/a&$-$&$+$&~&$k\sqrt n \pm n\log k+n\log n$\\
$n$&$n$&$\approx$&$+$&~&~&$kn \pm n\log k+n\log n$\\
\hline
$\log n$&$n\log n$&~&n/a&$-$&$+$&$k\log n \pm (n\log n)\log k$\\
$\sqrt n$&$n\log n$&~&n/a&$+$&~&$k\sqrt n \pm (n\log n)\log k$\\
$n$&$n\log n$&$-$&$+$&$+$&~&$kn \pm (n\log n)\log k$\\
$n\log n$&$n\log n$&$+$&~&~&~&$k \pm \log k$\\
\hline
$\log n$&$n^2$&~&~&n/a&$\approx$&$k\log n \pm n^2\log k$\\
$\sqrt n$&$n^2$&~&~&n/a&$+$&$k\sqrt n \pm n^2\log k$\\
$n$&$n^2$&~&n/a&$-$&$+$&$kn \pm n^2\log k$\\
$n\log n$&$n^2$&~&n/a&$\approx$&$+$&$k(n\log n) \pm n^2\log k$\\
$n\sqrt n$&$n^2$&n/a&$-$&$+$&~&$k(n\sqrt n) \pm n^2\log k$\\
$n^2$&$n^2$&$+$&~&~&~&$kn^2 \pm n^2\log k$\\
\hline
\end{tabular}
\caption{\label{fig:tableau}Running time comparison between the proposed algorithm and the adaptation of the algorithm from~\cite{BFJ03}. {\it The symbols $-$ (resp $+$, $\approx$) indicate the ranges of values for which our solution has a lower asymptotic complexity (resp. higher, of the same order). Empty cells at the right of a + (resp. left of a n/a) are filled with + (resp. n/a).}}
\end{table}

In summary, the table confirms that the proposed solution becomes more relevant as the difference between instant density and cumulated density increases, which is not surprising. It is also not surprising, given the presence of the factor $k$ versus $\log k$, that our solution is less efficient when the number of time steps increases. The table reveals some ranges of realistic values where the proposed solution behaves better than the other, for instance when the values of $\mu, m$, and $k$ are respectively $(O(n),\Theta(n^2),O(n))$; or $(O(\log n),\Omega(n\log n),O(n))$; or $(O(\sqrt n)), \Omega(n), O(\sqrt n))$.

Finally, the fact that the algorithm terminates as soon as temporal connectivity is satisfied allows us to put in perspective the impact of parameter $k$.


\end{document}